\documentclass[11pt]{article}
\usepackage{fullpage}
\usepackage{authblk}
\usepackage{graphicx}
\usepackage{float}
\usepackage{amsmath}
\usepackage{amsfonts}
\usepackage{amssymb}
\usepackage{amsthm}
\usepackage{mathrsfs}
\usepackage{bm}
\usepackage{mathtools}
\usepackage{enumitem}
\usepackage{algorithm}
\usepackage{algorithmic}

\usepackage{optidef}
\usepackage{IEEEtrantools}

\usepackage{url}
\usepackage{hyperref}

\usepackage{appendix}

\usepackage{caption}
\usepackage{subcaption}
\newtheorem{theorem}{Theorem}
\newtheorem{lemma}{Lemma}

\newtheorem{assumption}{Assumption}

\newtheorem*{claim*}{Claim}
\newtheorem*{remark*}{Remark}

\newtheorem*{definition*}{Definition}

\makeatletter

\usepackage{xparse}
\NewDocumentCommand{\lplabel}{o m}{%
  \makebox[0pt][r]{#2\hspace*{4em}}%
  \IfNoValueF{#1}
    {\def\@currentlabel{#2}\ltx@label{#1}}
}

\renewcommand\section{%
  \@startsection{section}{1}
                {\z@}%
                {-3.5ex \@plus -1ex \@minus -.2ex}%
                {2.3ex \@plus.2ex}%
                {\large\bfseries}
}
\renewcommand\subsection{%
  \@startsection{subsection}{2}
                {\z@}%
                {-3.25ex\@plus -1ex \@minus -.2ex}%
                {1sp}
                {\normalsize\bfseries}
}
\renewcommand\subsubsection{%
  \@startsection{subsubsection}{3}
                {\z@}%
                {-3.25ex\@plus -1ex \@minus -.2ex}%
                {1sp}
                {\normalfont\normalsize}
}

\def\IEEElabelanchoreqn#1{\bgroup
\def\@currentlabel{\p@equation\theequation}\relax
\def\@currentHref{\@IEEEtheHrefequation}\label{#1}\relax
\Hy@raisedlink{\hyper@anchorstart{\@currentHref}}\relax
\Hy@raisedlink{\hyper@anchorend}\egroup}
\makeatother
\newcommand{\subnumberinglabel}[1]{\IEEEyesnumber
  \IEEEyessubnumber*\IEEElabelanchoreqn{#1}}

\makeatother

\title{{\Large\bf  On the Approximate Core and Nucleon of\\ Flow Games with Public Arcs}\thanks{Supported in part by the National Natural Science Foundation of China (Nos.\,12001507, 11871442 and 12171444) and Natural Science Foundation of Shandong Province (No.\,ZR2020QA024).
}}

\author{Pengfei Liu, Han Xiao\thanks{Corresponding author. Email: hxiao@ouc.edu.cn.}, Tianhang Lu, Qizhi Fang}

\affil{School of Mathematical Sciences\\Ocean University of China\\Qingdao, China}

\date{\today}

\begin{document}


\maketitle

\openup 1.2\jot


\begin{abstract}
We investigate flow games featuring both private arcs owned by individual players and public arcs accessible cost-free to all coalitions. 
We explore two solution concepts within this framework: the approximate core and the nucleon.
The approximate core relaxes core requirements by permitting a bounded relative payoff deviation for every coalition,
and the nucleon is a multiplicative analogue of Schmeidler's nucleolus which lexicographically maximizes the vector consisting of relative payoff deviations for every coalition arranged in a non-decreasing order.
By leveraging a decomposition property for paths and cycles in a flow network, we derive complete characterizations for the approximate core and demonstrate that the nucleon can be computed in polynomial time.
\hfill



\hfill

\end{abstract}

\section{Introduction}

Cooperative game theory provides a powerful framework for understanding how individuals can strategically collaborate to achieve collective benefits.
Cooperations in network-based scenarios are often studied through flow games, where commodities or services are commonly perceived as flows within a network.
Flow games typically assume that each arc is owned by a single player \cite{GG92,KZ82a,KZ82b}.
Players collaborate strategically to maximize their individual revenues by  forming coalitions and routing flows along the network.
However, this assumption falls short of capturing the ubiquitous presence of shared resources, public goods, and common infrastructure in real-world networks.
Individual control and shared access often collaborate in various scenario.
For example, transportation networks with public roads and private tollways, or communication networks with private and shared exchange points for internet service providers.
To address this limitation, flow games with public arcs were introduced \cite{PRB06,RMPT96}.
As a generalized model, there are both private and public arcs, where private arcs are owned by individual players and public arcs are accessible to all players at no cost.
This extension enables a more realistic representation of scenarios where players have access to shared resources or infrastructure.
However, it also introduces new difficulties, prompting us to re-examine fundamental concepts within the framework of cooperative game theory. 

In cooperative game theory, the core is a central solution concept for analyzing the stability of cooperation.
It is the set of allocations where no coalition has an incentive to deviate from the grand coalition and operate independently.
However, the requirements of the core can be overly demanding, resulting in an empty core in many games \cite{BB20,Vazi22,XF23}.
Furthermore, even when the core is non-empty, finding a core allocation can be computationally challenging \cite{FZCD02}.
These limitations motivate the exploration of alternative solution concepts that relax strict requirements of the core.
Such ideas have been explored in early works \cite{Bond63, SS66}.
One such relaxation is the approximate core, which recognizes that full coalition stability is rare in practice and that minor compromises are often necessary.
It offers computational advantages over core existence checks, and, quantifies a game's proximity to admitting a core. 
More precisely, the approximate core allows coalitions to receive slightly less than their full potential, accepting deviations within a prescribed bound.
This flexibility reflects real-word factors like the costs of forming new coalitions, uncertainty about potential gains, or the willingness to settle for a “good enough” outcome.
The approximate core has two primary variants: the multiplicative one and additive one, which differ in how they measure deviations.
The multiplicative approximate core considers deviations relative to the full potential of every coalition.
It allows a coalition to receive slightly less than its full potential, but the shortfall is proportional to its worth.
The additive approximate core uses a fixed deviation tolerance for every coalition.
It allows every coalition to fall short of its full potential by up to a fixed amount.
When deviations shrink towards zero, both of them converge towards the core (if it exists), providing a smooth transition between the approximate core and the core.
The multiplicative and additive approximate cores have both been studied extensively.
The multiplicative one has been studied in \cite{FK93,FK98,GS04,Vazi22},
while the additive approximate core has been analyzed in \cite{AM85,BB20,MPS79,Schm69,XF23}, often serving as an intermediary concept in studying Schmeidler's nucleolus \cite{Schm69}.
The nucleolus is a key solution concept in cooperative game theory, which is the unique allocation that lexicographically maximizes the vectors consisting of absolute deviations for each coalition arranged in a non-decreasing order.
It always lies within the smallest non-empty additive approximate core.
In contrast to the nucleolus, the nucleon consists of allocations that lexicographically maximize the vectors consisting of proportional deviations for each coalition arranged in a non-decreasing order.
It was introduced by Faigle et al. \cite{FKFH98} and has been investigated in matching games \cite{FKFH98} and flow games without public arcs \cite{KP09}.
Analogous to the nucleolus, the nucleon is closely related to the multiplicative approximate core, providing another perspective on the balance between fairness and stability.
Note that the multiplicative approximate core is more sensitive to the worth of a coalition, while the additive version treats all coalitions equally.
The choice among these concepts depends on the specific situations and the desired fairness criteria.

In this work, we study the multiplicative approximate core and the nucleon in the flow game with public arcs.
For brevity, we will refer to the multiplicative approximate core as the approximate core in our following discussion.
Previous research on this model has investigated the existence and characterization of the core \cite{RMPT96} as well as the computation of the nucleolus \cite{PRB06}.
By exploiting properties of the maximum partially disjoint path problem, we examine the existence and characterization of the approximate core and show that the nucleon can be computed in polynomial time.
Our results complete the research line of flow games with public arcs.

The remainder of this paper is organized as follows.
Section \ref{sec:preliminaries} introduce basic notions and notations.
Besides, Section \ref{sec:preliminaries} introduces a auxiliary game model to study the approximate core and nucleon.
Section \ref{sec:core.auxiliary} studies the core of the auxiliary flow game.
Section \ref{sec:nucleon.auxiliary} studies the nucleon computation of the auxiliary flow game.
Section \ref{sec:main.results} presents the main results of this paper.

\section{Preliminaries}
\label{sec:preliminaries}

\subsection{Cooperative game theory}

A \emph{cooperative revenue game} $\Gamma=(N,\gamma)$ consists of a \emph{player se}t $N$ and a \emph{characteristic function} $\gamma:2^{N}\rightarrow \mathbb{R}$ with $\gamma(\emptyset)=0$.
We call $N$ the \emph{grand coalition} and call every subset $S\subseteq N$ a \emph{coalition}.
For any $\boldsymbol{x}\in \mathbb{R}^{N}$ and $S\subseteq N$, we use $x(S)$ to denote $\sum_{i\in S}x_{i}$.
An \emph{allocation} of $\Gamma$ is a vector $ \boldsymbol{x}\in \mathbb{R}^{N}_{\geq 0}$ with $x(N)=\gamma (N)$.
We use $\chi(\Gamma)$ to represent the set of allocations of $\Gamma$. 
The relative excess of an allocation is key concept in defining both the approximate core and the nucleon.
Given an allocation $\boldsymbol{x}\in \chi(\Gamma)$, its \emph{relative excess} for coalition $S\subseteq N$, denoted by $\text{re}_{\boldsymbol{x}}(S)$, is defined by $\frac{x(S)-\gamma (S)}{\gamma(S)}$ when $\gamma(S)>0$ and by $+\infty$ when $\gamma(S)=0$.

Now we define the approximate core.
The $\epsilon$-\emph{approximate core} of $\Gamma$, denoted by $\mathcal{C}_{\epsilon}(\Gamma)$, is the set of allocations whose relative excess is no less than $\epsilon$ for any $S\subseteq N$,
i.e., $\mathcal{C}_{\epsilon}(\Gamma)=\{\boldsymbol{x}\in \chi(\Gamma): x(S)\geq (1+\epsilon) \gamma (S), \forall S\subseteq N\}$.
Let $\epsilon^*=\max\{\epsilon: \mathcal{C}_{\epsilon}(\Gamma)\neq \emptyset\}$.
We call $\mathcal{C}_{\epsilon^*}(\Gamma)$ the \emph{optimal approximate core} of $\Gamma$.
By definition, $\epsilon^{*}\leq 0$.
When $\epsilon^*=0$, $\mathcal{C}_{\epsilon^*}(\Gamma)$ is the \emph{core} of $\Gamma$ and denoted by $\mathcal{C}(\Gamma)$, i.e.,
$\mathcal{C}(\Gamma)=\{\boldsymbol{x}\in \chi(\Gamma): x(S)\geq \gamma (S), \forall S\subseteq N\}$.

Next we define the nucleon.
Given an allocation  $\boldsymbol{x}\in \chi(\Gamma)$, its \emph{excess vector} $\theta (\boldsymbol{x})$ is a $(2^{\lvert N\rvert}-2)$-dimensional vector of all relative excesses $\text{re}_{\boldsymbol{x}}(S)$ arranged in a non-decreasing order, where $S\in 2^N\backslash \{\emptyset,N\}$.
The \emph{nucleon} of $\Gamma$, denoted by  $\eta (\Gamma)$, is the set of allocations that lexicographically maximize $ \theta(\boldsymbol{x}) $ over $\chi(\Gamma)$, i.e., 
$ \eta(\Gamma)=\{\boldsymbol{x}\in \chi(\Gamma):\theta(\boldsymbol{x})\succeq_{l}\theta(\boldsymbol{x}'), \forall \boldsymbol{x}'\in \chi(\Gamma)\}$.
Faigle et al. \cite{FKFH98} show that $\eta (\Gamma)$ can be computed by recursively solving the following sequential LPs in \eqref{eq:nucleon.lp} where $k\geq 0$. 
\begin{maxi!}|s|
  {}{\epsilon \label{eq:nucleon.lp.obj}}
  {\label{eq:nucleon.lp}}{LP_{k+1}:\quad}
  \addConstraint {x(S)}{\geq (1+\epsilon) \gamma (S),}{\quad \forall S\in 2^N\backslash  \textnormal{fix}(X_{k})}{\label{eq:nucleon.lp.c1}}
  \addConstraint {\boldsymbol{x}}{\in X_k .}{}{\label{eq:nucleon.lp.c2}}
\end{maxi!}
Initially, $\epsilon_0 =0$, $X_0=\chi (\Gamma)$, and $\text{fix}(X_0)=\{\emptyset,N\}$.
After solving $LP_{k+1}$, let $\epsilon_{k+1}$ be the optimum of $LP_{k+1}$, $X_{k+1}$ be the set of allocations $\boldsymbol{x}_{k+1}$ such that $(\boldsymbol{x}_{k+1},\epsilon_{k+1})$ is an optimal solution to $LP_{k+1}$, and $\textnormal{fix}(X_{k+1})=\{S\subseteq N: x(S)=x_{k+1}(S), \forall \boldsymbol{x} \in X_{k+1}\}$ be the set of \emph{fixed coalitions} in $X_{k+1}$.
Note that $\epsilon_{k+1} > \epsilon_k$ and $X_{k+1}\subseteq X_k$. 
When $X_{k+1}$ becomes a singleton or $\gamma (S)=0$ for $S\in 2^N\backslash \textnormal{fix}(X_{k+1})$, the recursion ends and $X_{k+1}=\eta (\Gamma)$.
Since the dimension of $X_{k+1}$ is less than that of $X_k$, it takes up to $\lvert N\rvert$ rounds before $X_{k+1}$ converges to the nucleon $\eta (\Gamma)$.
The main difficulty on the computation of $\eta (\Gamma)$ lies in determining $X_{k+1}$, $\textnormal{fix}(X_{k+1})$ and verifying exponential number of constraints in $LP_{k+1}$.

\subsection{Flow game}

Let $D=(V,E;s,t)$ be a \emph{network with unit arc capacities}, where $V$ is the \emph{vertex set}, $E$ is the \emph{arc set}, $s$ is the \emph{source} and $t$ is the \emph{sink}.
Parallel arcs are allowed in $D$.
We use $\Gamma_D=(N,\gamma)$ to denote the \emph{flow game with public arcs} defined on $D=(V,E;s,t)$.
Each player in $N$ possesses a distinct \emph{private arc} in $E$.
For brevity, we can view the player set $N$ as a subset of the arc set $E$.
The arcs in $M=E\backslash N$ are called \emph{public arcs}, which are not owned by any specific player and can be used for free by any coalition.
For each coalition $S\subseteq N$, $\gamma (S)$ equals the maximum flow value in the induced network $D_S =(V,S\cup M;s,t)$.

Given that $D=(V,E;s,t)$ is a network with unit arc capacities, we only consider integer-valued flows without circulations.
If a flow contains a circulation, we can always diminish the circulation without affecting the value of the flow.
According to the flow decomposition theorem \cite{Schr03}, we define a \emph{maximum flow} of $D$ as a maximum set of arc disjoint $s$-$t$ paths in $D$.
A \emph{path} is a set of arcs joining a sequence of distinct vertices along the same direction.
A path is a \emph{subpath} of another path if it is a contiguous subsequence of the latter.
We use $\mathscr{P}$ to denote the family of all $s$-$t$ paths in $D$.
A \emph{minimum $s$-$t$ cut} of $D$ is a minimum arc set that intersect every $s$-$t$ path in $D$.

In the study of flow games, the following assumptions are commonly made \cite{KP09,RMPT96}.
Assumption \ref{assumption:1} guarantees that the game $\Gamma_D=(N,\gamma)$ is \emph{normalized}, i.e., $\gamma(\emptyset)=0$.
Assumption \ref{assumption:2} removes redundant arcs in the network $D$.

\begin{assumption}\label{assumption:1}
 Every $s$-$t$ path of $D$ contains an arc from $N$.
\end{assumption}
\begin{assumption}\label{assumption:2}
  Every arc of $D$ belongs to an $s$-$t$ path.
\end{assumption}

For the flow game $\Gamma_D$ with player set $N=E$,
Kalai and Zemel \cite{KZ82a,KZ82b} show that the core is always non-empty for every subgame and provide structural characterizations for the core;
Moreover,
Kern and Paulusma \cite{KP09} give an efficiently algorithm for computing the nucleon.
For the flow game $\Gamma_D$ with player set $N\subseteq E$, Reijnierse et al. \cite{RMPT96} study the existence and structural characterizations of the core;
And we will settle the characterization of the approximate core and the computation of the nucleon in this paper.

We study flow games with public arcs by resorting to the \emph{maximum partially disjoint path problem}.
Let $F\subseteq E$ be an arc subset of $D$.
Two paths of $D$ are \emph{disjoint on $F$} if they have no common arc in $F$.
We use $\sigma_F$ to denote the maximum number of $s$-$t$ paths disjoint on $F$.
Clearly, $\sigma_F \geq \sigma_E $, since paths disjoints on $E$ are naturally disjoint on $F$.
A cut of $D$ is \emph{constrained to $F$} if it is a subset of $F$.
A min-max relation for partially disjoint paths and constrained cuts can be derived from Menger's theorem.

\begin{lemma}[Schrijver \cite{Schr03}]\label{thm:partial.disjoint.min-max}
  For any $F\subseteq E$,
  the maximum number of $s$-$t$ paths disjoint on $F$ is equal to the minimum size of an $s$-$t$ cut constrained to $F$.
\end{lemma}

Throughout this paper, we pick an arbitrary maximum set $\mathscr{B}$ of $s$-$t$ paths disjoint on $N$ to characterize all $s$-$t$ paths and all cycle in $D$.
A $u$-$v$ \emph{jump} is a $u$-$v$ path where both $u$ and $v$ lie on the paths from $\mathscr{B}$ but no intermedia vertices belong to paths from $\mathscr{B}$.
Every $u$-$v$ jump determines a \emph{jump pair} $(u,v)$.
We use $\mathscr{J}$ and $J$ to denote the set of all jumps and all jump pairs over $\mathscr{B}$, respectively.
Notice that $J$ can be determined by solving $O(\lvert V\rvert)^2$ shortest path problems.

\begin{lemma}\label{thm:jump.composition}
  Every $s$-$t$ path and every cycle is a composition of subpaths of paths from $\mathscr{B}$ and jumps from $\mathscr{J}$.
\end{lemma}

Lemma \ref{thm:jump.composition} is trivial.
To see this, let $W$ be an $s$-$t$ path or a cycle in $D$.
We use $*$ for the concatenation of paths.
Then $W$ admits a decomposition 
\begin{equation}\label{eq:decomposition}
  P_0 * Q_1 * P_1 * \cdots * Q_r * P_r,
\end{equation}
where $P_i$ is a subpath of paths from $\mathscr{B}$ and $Q_i$ is a jump from $\mathscr{J}$.
Note that some $P_i$ in \eqref{eq:decomposition} may be empty.
For brevity, let $\mathscr{J}_W$ denote the set of all jumps in the decomposition \eqref{eq:decomposition} of $W$ and $J_W$ denote the set of all jump pairs defined from $\mathscr{J}_W$.

\begin{lemma}\label{thm:jump.uniqueness}
  For any jump $Q\in \mathscr{J}$, there exists either an $s$-$t$ path $P$ with $\mathscr{J}_P =\{Q\}$ or a cycle $C$ with $\mathscr{J}_C =\{Q\}$.
\end{lemma}
\begin{proof}
  Let $Q\in \mathscr{J}$ be a $u$-$v$ jump.
  Let $P^u$ and $P^v$ be the $s$-$t$ paths from $\mathscr{B}$ passing $u$ and $v$ respectively.
  Denote by $P^u_{s-u}$ and $P^v_{v-t}$ the $s$-$u$ path and $v$-$t$ paths along $P^u$ and $P^v$ respectively.
  If $P^u_{s-u}$ and $P^v_{v-t}$ share no common vertex, then $P^u_{s-u} * Q * P^v_{v-t}$ is a desired $s-t$ path.
  Otherwise, let $w$ be the last vertex on $P^u_{s-u}$ shared with $P^v_{v-t}$.
  Denote by $P^u_{w-u}$ and $P^v_{v-w}$ the $w$-$u$ path and $v$-$w$ path along $P^u$ and $P^v$ respectively.
  Then $P^u_{w-u} * Q * P^v_{v-w}$ is a desired cycle.
\end{proof}

\subsection{Auxiliary game}
To study flow game $\Gamma_D =(N,\gamma)$, we introduce its \emph{auxiliary game} $\widetilde{\Gamma}_D=(N,\tilde{\gamma})$.
The player set of $\widetilde{\Gamma}_D$ is $N$, the same with that of $\Gamma_D$.
The characteristic function of $\widetilde{\Gamma}_D$ is defined by $\tilde{\gamma}(N)=\frac{\sigma_N}{\sigma_E}\gamma (N)=\sigma_N$ and $\tilde{\gamma}(S)=\gamma (S)$ for $S\subset N$.
Faigle et al. \cite{FKFH98} show that the study of the approximate core and nucleon of $\Gamma_D$ can be conducted on $\widetilde{\Gamma}_D$.
For completeness, we include this result as Lemma \ref{thm:multiplication.factor}.
In our case, $\boldsymbol{x}$ belongs to the optimal approximate core of $\widetilde{\Gamma}_D$ if and only if $\frac{\sigma_E}{\sigma_N} \boldsymbol{x}$ belongs to the optimal approximate core of $\Gamma_D$, and $\boldsymbol{x}$ belongs to the nucleon of $\widetilde{\Gamma}_D$ if and only if $\frac{\sigma_E}{\sigma_N} \boldsymbol{x}$ belongs to the nucleon of $\Gamma_D$.
Therefore, we will mainly work on $\widetilde{\Gamma}_D$ and study its optimal approximate core and nucleon computation in our following discussion.
Moreover, we further assume that $\sigma_N \geq 2$, as we shall see in Section \ref{sec:main.results} that when $\sigma_N =1$, the nucleon and core coincide in $\Gamma_D$.

\begin{lemma}[Faigle et al. \cite{FKFH98}]\label{thm:multiplication.factor}
  The optimal approximate core and nucleon of a game is changed by a multiplication factor $\alpha$ ($\alpha>0$) if we multiply the grand coalition value by $\alpha$ while keeping the other coalition values intact.
\end{lemma}

\section{The core of $\widetilde{\Gamma}_D$}
\label{sec:core.auxiliary}

In this section, we study the core of $\widetilde{\Gamma}_D$.
We first show that $\mathcal{C} (\widetilde{\Gamma}_D)$ is non-empty and provide two structural characterizations for $\mathcal{C} (\widetilde{\Gamma}_D)$.

\begin{lemma}\label{thm:nucleon.lp.1.equiv.opt.double.descr}
  $\mathcal{C} (\widetilde{\Gamma}_D) = \{ \boldsymbol{x}\in \chi(\widetilde{\Gamma}_D): x(N\cap P)\geq 1, \forall P\in \mathscr{P} \} \neq \emptyset$.
\end{lemma}
\begin{proof}
  Clearly, $\mathcal{C} (\widetilde{\Gamma}_D) \subseteq \{ \boldsymbol{x}\in \chi(\widetilde{\Gamma}_D): x(N\cap P)\geq 1, \forall P\in \mathscr{P} \}$.
  We prove the reverse inclusion in the following.
  Let $\boldsymbol{x}\in \chi(\widetilde{\Gamma}_D)$ be an allocation such that $x(N\cap P)\geq 1$ for any $P\in \mathscr{P}$.
  We show that $\boldsymbol{x}\in \mathcal{C} (\widetilde{\Gamma}_D)$.
  Let $S\subseteq N$.
  We may assume that $\tilde{\gamma}(S)\geq 1$, otherwise $x(S)\geq 0=\tilde{\gamma}(S)$ is trivial.
  Let $\mathcal{P}_S$ be a maximum set of disjoint $s$-$t$ paths in $D_S$.
  It follows that $x(S)\geq \sum_{P\in \mathcal{P}_S} x(N\cap P)=\tilde{\gamma}(S)$.
  Hence we have $\boldsymbol{x}\in \mathcal{C} (\widetilde{\Gamma}_D)$.

  It remains to show that $\mathcal{C} (\widetilde{\Gamma}_D)\neq \emptyset$.
  By Lemma \ref{thm:partial.disjoint.min-max}, the incidence vector of every minimum $s$-$t$ cut constrained to $N$ is an allocation in $\mathcal{C} (\widetilde{\Gamma}_D)$.
  Hence $\mathcal{C} (\widetilde{\Gamma}_D)$ is always non-empty.
\end{proof}

\begin{lemma}\label{thm:core.v.rep}
  $\mathcal{C} (\widetilde{\Gamma}_D)$ is the convex hull of incidence vectors of minimum $s$-$t$ cuts constrained to $N$.
\end{lemma}
\begin{proof}
  The lemma follows from Lemmas \ref{thm:partial.disjoint.min-max} and \ref{thm:nucleon.lp.1.equiv.opt.double.descr}.
\end{proof}

The rest of this section is devoted to properties of $\mathcal{C} (\widetilde{\Gamma}_D)$, which are crucial in the nucleon computation of $\widetilde{\Gamma}_D$.

\begin{lemma}\label{thm:core.allocation.non-veto}
  For any $\boldsymbol{x}\in \mathcal{C} (\widetilde{\Gamma}_D)$, $x(e)=0$ if there exists a maximum set of $s$-$t$ paths disjoint on $N$ where every path avoids $e\in N$.
\end{lemma}
\begin{proof}
  Let $e\in N$ be an arc such that there exists a maximum set of $s$-$t$ paths disjoint on $N$ where every path avoids it.
  It follows that $e$ does not belong to any minimum $s$-$t$ cut constrained to $N$.
  By Lemma \ref{thm:core.v.rep}, we have $x(e)=0$ for any $\boldsymbol{x}\in \mathcal{C} (\widetilde{\Gamma}_D)$.
\end{proof}

Motivated by work \cite{PRB06}, we introduce the potentials on vertices.
For each $\boldsymbol{x}\in \chi(\widetilde{\Gamma}_D)$, its corresponding \emph{potential} $\boldsymbol{\phi}_{\boldsymbol{x}} \in [0,1]^V$ is defined as follows:
first construct an auxiliary graph $D^{\boldsymbol{x}}$ from $D$ where the arc length is $x(e)$ if $e\in N$ and $0$ otherwise;
then for $v\in V$, let $\phi_{\boldsymbol{x}} (v)$ be the length of the shortest $s$-$v$ path in $D^{\boldsymbol{x}}$.
Clearly, the potential $\boldsymbol{\phi}_{\boldsymbol{x}}$ for $\boldsymbol{x}\in \chi(\widetilde{\Gamma}_D)$ can be determined efficiently.

\vspace{-10pt}
\begin{IEEEeqnarray}{rCl'l}
  \subnumberinglabel{eq:core.potential}
  x(e) +\phi(u) -\phi(v) & \geq & 0, & \forall e=(u,v)\in N, \label{eq:core.potential.c1} \\
  \phi(u) - \phi(v) & \geq & 0, & \forall e=(u,v)\in M, \label{eq:core.potential.c2}\\
  \phi(s) & = & 0, \label{eq:core.potential.c3}\\
  \phi(t) & = & 1. \label{eq:core.potential.c4}
\end{IEEEeqnarray}

\begin{lemma}\label{thm:potential.property}
  For $\boldsymbol{x}\in \mathcal{C}(\widetilde{\Gamma}_D)$ and its potential $\boldsymbol{\phi}_{\boldsymbol{x}}$, $(\boldsymbol{x},\boldsymbol{\phi}_{\boldsymbol{x}})$ satisfies \eqref{eq:core.potential}.
  Moreover, for $e\in E$ such that there exists a maximum set of $s$-$t$ paths disjoint on $N$ with one path using it, equality in \eqref{eq:core.potential.c1} holds if $e\in N$ and equality in \eqref{eq:core.potential.c2} holds otherwise.
\end{lemma}

\begin{proof}
  Let $\boldsymbol{x}\in \mathcal{C}(\widetilde{\Gamma}_D)$ and $\boldsymbol{\phi}_{\boldsymbol{x}}$ be its potential.
  By definition of $\boldsymbol{\phi}_{\boldsymbol{x}}$, \eqref{eq:core.potential.c3} is trivial.
  Furthermore, \eqref{eq:core.potential.c1} and \eqref{eq:core.potential.c2} follow from triangle inequality.
  By Lemma \ref{thm:nucleon.lp.1.equiv.opt.double.descr}, $x(N\cap P)\geq 1$ for $P\in \mathscr{P}$.
  Recall that $\mathscr{B}$ be a maximum set of $s$-$t$ paths disjoint on $N$.
  Then $\sigma_N =x(N)\geq \sum_{P\in \mathscr{B}} x(N\cap P)\geq \sigma_N$ implies that $x(N\cap P)=1$ for $P\in \mathscr{B}$.
  Hence \eqref{eq:core.potential.c4} follows.
  For $P\in \mathscr{P}$, \eqref{eq:core.potential} implies that
  \begin{equation}\label{eq:potential.path}
    x(N \cap P) \geq \sum_{(u,v)\in N\cap P}\big(\phi(v) - \phi(u)\big) \geq \sum_{(u,v)\in P}\big(\phi(v) - \phi(u)\big) = \phi (t)-\phi (s)=1.
  \end{equation}
  Let $P^* \in \mathscr{B}$.
  Recall that $x(N\cap P^*)=1$.
  Then \eqref{eq:potential.path} implies that $x(e)+\phi_{\boldsymbol{x}} (u) - \phi_{\boldsymbol{x}} (v)=0$ for $e\in N\cap P^*$ and $\phi_{\boldsymbol{x}} (u) - \phi_{\boldsymbol{x}} (v)=0$ for $e\in M\cap P^*$.
\end{proof}

By Lemma \ref{thm:potential.property}, potentials provide an alternative characterization for the core.

\begin{lemma}\label{thm:core.potential}
  $\mathcal{C}(\widetilde{\Gamma}_D)=\{\boldsymbol{x}\in \chi (\widetilde{\Gamma}_{D}): (\boldsymbol{x},\boldsymbol{\phi}_{\boldsymbol{x}}) \text{ satisfies }\eqref{eq:core.potential}\}$.
\end{lemma}

We conclude this section by showing that every $s$-$t$ path and every cycles can be characterized with core allocations and corresponding potentials.

\begin{lemma}\label{thm:potential.rep.path-cycle}
  For $\boldsymbol{x}\in \mathcal{C}(\widetilde{\Gamma}_D)$ and its potential $\boldsymbol{\phi}_{\boldsymbol{x}}$, the following properties hold.
  \begin{enumerate}[label*=\roman*\emph{)}]
    \item\label{itm:potential.rep.path} For any $s$-$t$ path $P$, $x(N\cap P)=1+\sum_{(u,v)\in J_P} \big(\phi_{\boldsymbol{x}}(u)-\phi_{\boldsymbol{x}}(v)\big)$.
    \item\label{itm:potential.rep.cycle} For any cycle $C$, $x(N\cap C)= \sum_{(u,v)\in J_C} \big(\phi_{\boldsymbol{x}}(u)-\phi_{\boldsymbol{x}}(v)\big)$.
  \end{enumerate}
\end{lemma}
\begin{proof}
  Let $W$ be an $s$-$t$ path or a cycle in $D$.
  Recall that $W$ can be decomposed into $P_0 * Q_1 * P_1 * \cdots * Q_r * P_r$, where $P_i$ is a $v_i$-$u_{i+1}$ subpath of paths from $\mathscr{B}$ and $Q_i$ is a $u_i$-$v_i$ jump from $\mathscr{J}$.
  Note that $v_0=s$ and $u_{r+1}=t$ if $W$ is an $s$-$t$ path and $v_0=u_{r+1}$ if $W$ is a cycle.
  By Lemma \ref{thm:core.allocation.non-veto}, $x(N\cap Q_i)=0$.
  It follows from Lemma \ref{thm:potential.property} that
  \begin{equation*}
    x(N\cap W) = \sum_{i=1}^{r} x(N\cap P_i) = \sum_{i=0}^{r} \big(\phi_{\boldsymbol{x}}(u_{i+1})-\phi_{\boldsymbol{x}}(v_{i})\big) = \phi_{\boldsymbol{x}}(u_{r+1}) - \phi_{\boldsymbol{x}}(v_0) +\sum_{i=1}^{r}\big(\phi_{\boldsymbol{x}}(u_i)-\phi_{\boldsymbol{x}}(v_i)\big).
  \end{equation*}
  Hence $x(N\cap W)=1+\sum_{i=1}^{r}\big(\phi_{\boldsymbol{x}}(u_i)-\phi_{\boldsymbol{x}}(v_i)\big)$ if $W$ is an $s$-$t$ path and $x(N\cap W)=\sum_{i=1}^{r}\big(\phi_{\boldsymbol{x}}(u_i)-\phi_{\boldsymbol{x}}(v_i)\big)$ if $W$ is a cycle.
\end{proof}

\section{The nucleon of $\widetilde{\Gamma}_D$}
\label{sec:nucleon.auxiliary}

In this section, we show that the nucleon of $\widetilde{\Gamma}_D$ can be computed efficiently.
The nucleon of $\widetilde{\Gamma}_D$ can be computed by recursively solving a sequence of linear programs as specified in \eqref{eq:auxiliary.nucleon.lp}.
\begin{maxi!}|s|
  {}{\epsilon \label{eq:auxiliary.nucleon.lp.obj}}
  {\label{eq:auxiliary.nucleon.lp}}{\widetilde{LP}_{k+1}:\quad}
  \addConstraint {x(S)}{\geq (1+\epsilon) \widetilde{\gamma} (S),}{\quad \forall S\in 2^N\backslash  \textnormal{fix}(\widetilde{X}_{k})}{\label{eq:auxiliary.nucleon.lp.c1}}
  \addConstraint {\boldsymbol{x}}{\in \widetilde{X}_k .}{}{\label{eq:auxiliary.nucleon.lp.c2}}
\end{maxi!}
Initially, $\epsilon_0 =0$ and $\widetilde{X}_0=\chi (\widetilde{\Gamma})$.
After solving $\widetilde{LP}_{k+1}$, let ${\epsilon}_{k+1}$ be the optimum and $\widetilde{X}_{k+1}$ be the set of vectors $\boldsymbol{x}_{k+1}$ such that $(\boldsymbol{x}_{k+1},{\epsilon}_{k+1})$ is an optimal solution.
By Lemma \ref{thm:nucleon.lp.1.equiv.opt.double.descr}, we have the following results for $\widetilde{LP}_1$.

\begin{lemma}\label{thm:auxiliary.nucleon.lp.1}
${\epsilon}_1 =0$ and $\widetilde{X}_1 =\mathcal{C}(\widetilde{\Gamma}_D)$.
\end{lemma}


\begin{lemma}\label{thm:separation.oracle.opt.1}
  The separation problem of $\widetilde{X}_1$ can be solved in polynomial time. 
\end{lemma}
\begin{proof}
  For $\boldsymbol{x}\in \widetilde{X}_0$, construct an auxiliary graph $D^{\boldsymbol{x}}$ from $D$ where the arc length is $x(e)$ if $e\in N$ and $0$ otherwise.
  By Lemmas \ref{thm:nucleon.lp.1.equiv.opt.double.descr} and \ref{thm:auxiliary.nucleon.lp.1}, $\boldsymbol{x}\in \widetilde{X}_1$ if the shortest $s$-$t$ path in $D^{\boldsymbol{x}}$ has length no less than $1$, and a violated constraint of $\widetilde{X}_1$ is determined by the shortest $s$-$t$ path in $D^{\boldsymbol{x}}$ otherwise.
\end{proof}

Now assume that $k\geq 1$, and the separation problem of $\widetilde{X}_k$ can be solved in polynomial time.
We employ induction to prove that the separation problem of $\widetilde{X}_{k+1}$ can be solved in polynomial time.
For $S\in 2^N \backslash \{\emptyset,N\}$, let $\mathcal{P}_S$ denote a maximum set of disjoint $s$-$t$ paths in $D_S$ and $P_S$ denote the arc set of paths from $\mathcal{P}_S$.
By Lemma \ref{thm:auxiliary.nucleon.lp.1}, $\widetilde{X}_k\subseteq \mathcal{C}(\widetilde{\Gamma}_D)$.
It follows from Lemma \ref{thm:potential.rep.path-cycle} that for any $\boldsymbol{x}\in \widetilde{X}_k$ and its potential $\boldsymbol{\phi}_{\boldsymbol{x}}$, we have
\begin{equation}\label{eq:auxiliary.nucleon.lp.equiv.c1.alt}
  x(S)-\widetilde{\gamma}(S)= x(S\backslash P_S)+\sum_{P\in \mathcal{P}_S}\sum_{(u,v)\in J_P} \big(\phi_{\boldsymbol{x}}(u)-\phi_{\boldsymbol{x}}(v)\big).
\end{equation}
Moreover, $x(S)-\widetilde{\gamma}(S)$ is determined by the value $x(e)$ for $e\in S\backslash P_S$ and the potential difference $\phi_{\boldsymbol{x}}(u)-\phi_{\boldsymbol{x}}(v)$ for $(u,v)\in \cup_{P\in \mathcal{P}_S} J_P$.
In particular, if $S\in 2^N \backslash \text{fix}(\widetilde{X}_k)$, there exist either unfixed private arcs in $S\backslash P_S$ or unfixed jump pairs in $\cup_{P\in \mathcal{P}_S} J_P$.
Before proceeding, we introduce $N_k$ and $J_k$ to denote the set of private arcs and jump pairs that are fixed in $\widetilde{X}_{k}$, respectively.
\begin{itemize}
  \item Let $N_{k}$ be the set of arcs from $N$ that are fixed in $\widetilde{X}_{k}$, i.e., $N_{k}=\{e\in N: x' (e) = x'' (e), \forall \boldsymbol{x}',\boldsymbol{x}''\in \widetilde{X}_{k}\}$.
  Clearly, $N_k \subseteq N_{k+1}$.
  \item Let $J_{k}$ be the set of jump pairs $(u,v)$ from $J$ whose potential difference are fixed in $\widetilde{X}_{k}$, i.e., $J_k=\{(u,v)\in J:\phi_{\boldsymbol{x}'} (u) -\phi_{\boldsymbol{x}'} (v)=\phi_{\boldsymbol{x}''} (u) - \phi_{\boldsymbol{x}''} (v), \forall \boldsymbol{x}',\boldsymbol{x}''\in \widetilde{X}_{k}\}$.
  Clearly, $J_{k} \subseteq J_{k+1}$.
\end{itemize}

\begin{lemma}\label{thm:fixed.poly.solvable}
  Both $N_k$ and $J_k$ can be determined in polynomial time.
\end{lemma}
\begin{proof}
  For $e\in N$, we have $e\in N_k$ if $\max \{x(e):\boldsymbol{x}\in \widetilde{X}_{k}\}=\min \{x(e):\boldsymbol{x}\in \widetilde{X}_{k}\}$.
  For $(u,v)\in J$, we have $(u,v)\in J_k$ if $\max\{\phi_{\boldsymbol{x}} (u)-\phi_{\boldsymbol{x}} (v): \boldsymbol{x}\in \widetilde{X}_k\} = \min\{\phi_{\boldsymbol{x}} (u)-\phi_{\boldsymbol{x}} (v): \boldsymbol{x}\in \widetilde{X}_k\}$.
  By induction hypothesis,  the separation problem of $\widetilde{X}_k$ can be solved in polynomial-time.
  Hence both $N_k$ and $J_k$ can be determined in polynomial time.
\end{proof}

In order to examine the critical conditions at which the constraint \eqref{eq:auxiliary.nucleon.lp.equiv.c1.alt} goes fixed, we introduce an auxiliary network $D_k=(V,E_k;s,t)$ for $\widetilde{LP}_{k+1}$.
The network $D_k$ is obtained from $D=(V,E;s,t)$ by removing all arcs from $u$-$v$ jumps in $\mathscr{J}$ for $(u,v)\not\in J_k$.
Notice that for the arc set $F$ of an arbitrary disjoint $s$-$t$ paths in $D_k$, $N\cap F\in \mbox{fix}(\widetilde{X}_k)$ follows from \eqref{eq:auxiliary.nucleon.lp.equiv.c1.alt}.
In the following, we identify two important classes of unfixed coalitions in $\widetilde{X}_k$.
 
\begin{itemize}
  \item For unfixed private arc $e\in N\backslash N_k$, let $\mathcal{F}^{-}_k (e)$ be the collection of arc sets of every possible disjoint $s$-$t$ paths in $D_k -e$, and let $\mathcal{O}^e_k =\{(N\cap F)\cup \{e\}:F\in \mathcal{F}^{-}_k (e)\}$.
  For $\boldsymbol{x}', \boldsymbol{x}''\in \widetilde{X}_{k}$ and $O\in \mathcal{O}^e_k$, we have $x'(O)-x'(e) = x''(O)-x''(e)$, implying $O \backslash \{e\}\in \text{fix}(\widetilde{X}_k)$.
  Hence $\mathcal{O}^e_k$ consists of critical unfixed coalitions determined by $e\in N\backslash N_k$.
  \item For unfixed jump pair $(u,v)\in J \backslash J_k$, let $\mathcal{F}^{+}_k (u,v)$ be the collection of arc sets of every possible disjoint $s$-$t$ paths with one path containing $(u,v)$ in $D_k +(u,v)$, and let $\mathcal{R}^{(u,v)}_k =\{N\cap F:F\in \mathcal{F}^{+}_k (u,v)\}$.
  For $\boldsymbol{x}', \boldsymbol{x}''\in \widetilde{X}_{k}$ and $R\in \mathcal{R}^{(u,v)}_k$, we have $x'(R)-\big(\phi_{\boldsymbol{x}'}(u)-\phi_{\boldsymbol{x}'}(v)\big) = x''(R)-\big(\phi_{\boldsymbol{x}''}(u)-\phi_{\boldsymbol{x}''}(v)\big)$.
  Hence $\mathcal{R}^{(u,v)}_k$ consists of critical unfixed coalitions determined by $(u,v)\in J \backslash J_k$.
\end{itemize}

Now we introduce a linear program to solve the separation problem of $\widetilde{LP}_{k+1}$ defined in \eqref{eq:auxiliary.nucleon.lp}.
For $\boldsymbol{x}\in \widetilde{X}_k$, let  $\widetilde{LP}_{k+1}^{\boldsymbol{x}}$ be an associated linear program as specified in \eqref{eq:separation.oracle.opt.k+1}.
\begin{maxi!}|s|
  {}{\epsilon \label{eq:separation.oracle.opt.k+1.obj}}
  {\label{eq:separation.oracle.opt.k+1}}{\widetilde{LP}^{\boldsymbol{x}}_{k+1}:\quad}
  \addConstraint {x(O^{\boldsymbol{x}}_k)}{\geq (1+\epsilon) \widetilde{\gamma}(O^{\boldsymbol{x}}_k)}{}{\label{eq:separation.oracle.opt.k+1.c1}}
  \addConstraint {x(R^{\boldsymbol{x}}_k)}{\geq (1+\epsilon) \widetilde{\gamma}(R^{\boldsymbol{x}}_k),}{}{\label{eq:separation.oracle.opt.k+1.c2}}
\end{maxi!}
where 
\begin{equation*}
  O_k^{\boldsymbol{x}}=\mathop{\arg\min}_{O\in\mathcal{O}^e_k} \min_{e\in N\backslash N_k} \text{re}_{\boldsymbol{x}} (O),
\end{equation*}
and
\begin{equation*}
  R_k^{\boldsymbol{x}}=\mathop{\arg\min}_{R\in\mathcal{R}^{(u,v)}_k} \min_{(u,v)\in J\backslash J_k} \text{re}_{\boldsymbol{x}} (R).
\end{equation*}

\begin{lemma}
  For $\boldsymbol{x}\in \widetilde{X}_k$, both $O_k^{\boldsymbol{x}}$ and $R_k^{\boldsymbol{x}}$ can be determined in polynomial time.
\end{lemma}
\begin{proof}
  Let $\boldsymbol{x}\in \widetilde{X}_k$ and $D^{\boldsymbol{x}}_k$ be a weighted graph obtained from $D_k$ where the arc weight is $x(e)$ if $e\in N\cap E_k$ and $0$ otherwise.

  Let $e\in N\backslash N_k$.
  Denote by $F^-_{\tau}$ the arc set of $\tau$ disjoint $s$-$t$ paths with minimum weight in $D^{\boldsymbol{x}}_k -e$.
  Note that $F^-_{\tau}$ can be determined efficiently as a minimum cost flow problem \cite{Schr03}.
  Moreover, $N\cap (F^-_{\tau} \cup \{e\})=\mathop{\arg\min}_{O\in \mathcal{O}^e_k} \text{re}_{\boldsymbol{x}} (O)$ for $\tilde{\gamma}(O)=\tau$.
  Hence $\mathop{\arg\min}_{O\in \mathcal{O}^e_k} \text{re}_{\boldsymbol{x}} (O)$ can be computed by solving at most $\sigma_E$ minimum cost flow problems.
  It follows that $O_k^{\boldsymbol{x}}$ can be computed in polynomial time.

  Let $(u,v)\in J\backslash J_k$ and $Q_{u-v}$ be the a $u$-$v$ jump.
  Let $F^+_{\tau}$ be the arc set of $\tau$ disjoint $s$-$t$ paths with minimum weight in $D^{\boldsymbol{x}}_k +(u,v)$, where the weight of newly added arc $(u,v)$ is $-\sigma_N -1$.
  Note that $F^+_{\tau}$ can also be determined efficiently as a minimum cost flow problem \cite{Schr03}.
  Moreover, $N\cap (F^{+}_{\tau}\cup Q_{u-v}) = \arg\min_{R\in \mathcal{R}^{(u,v)}_{k}} \text{re}_{\boldsymbol{x}} (R)$ for $\tilde{\gamma}(R)=\tau$.
  Hence $\mathop{\arg\min}_{R\in \mathcal{R}^{(u,v)}_k} \text{re}_{\boldsymbol{x}} (R)$ can be computed by solving at most $\sigma_E$ minimum cost flow problems.
  It follows that $R_k^{\boldsymbol{x}}$ can be computed in polynomial time.
\end{proof}

\begin{lemma}
  For $\boldsymbol{x}\in \widetilde{X}_k$ and $(u,v)\in J\backslash J_k$, we have $\phi_{\boldsymbol{x}} (u) - \phi_{\boldsymbol{x}} (v) \geq \min \{\text{\emph{re}}_{\boldsymbol{x}} (O_k^{\boldsymbol{x}}),\text{\emph{re}}_{\boldsymbol{x}} (R_k^{\boldsymbol{x}})\}$.
\end{lemma}
\begin{proof}
  Let $(u,v)\in J\backslash J_k$ and $Q$ be a $u$-$v$ jump.
  By Lemma \ref{thm:jump.uniqueness}, there exists either an $s$-$t$ path $P$ with $\mathscr{J}_P =\{Q\}$ or a cycle $C$ with $\mathscr{J}_C =\{Q\}$.

  For $s$-$t$ path $P$, Lemma \ref{thm:potential.rep.path-cycle} implies that $x(N\cap P)=1+\phi_{\boldsymbol{x}} (u) -\phi_{\boldsymbol{x}} (v)$.
  Hence $N\cap P\not\in \text{fix} (\widetilde{X}_k)$.
  Moreover, we have $N\cap P\in \mathcal{R}^{(u,v)}_k$, implying that $\text{re}_{\boldsymbol{x}} (N\cap P)= \phi_{\boldsymbol{x}} (u) -\phi_{\boldsymbol{x}} (v)\geq \text{re}_{\boldsymbol{x}} (R^{\boldsymbol{x}}_k)$.

  For cycle $C$, Lemma \ref{thm:potential.rep.path-cycle} implies that $x(N\cap C)=\phi_{\boldsymbol{x}} (u) -\phi_{\boldsymbol{x}} (v)$.
  Then $N\cap C\not\in \text{fix} (\widetilde{X}_k)$ suggests that $C$ contains an arc $e'\in N\backslash N_k$.
  Since $\sigma_N\geq 2$, there exists $P'\in \mathscr{B}$ avoiding $e'$.
  Moreover, Lemma \ref{thm:potential.rep.path-cycle} implies that $x\big( (N\cap P')\cup \{e'\}\big)=x(e')+1$.
  Note that $(N\cap P')\cup \{e'\}\in \mathcal{O}^{e'}_k$.
  It follows that $\phi_{\boldsymbol{x}} (u) -\phi_{\boldsymbol{x}} (v)=x(N\cap C)\geq x(e')=\text{re}_{\boldsymbol{x}} \big((N\cap P')\cup \{e'\}\big)\geq \text{re}_{\boldsymbol{x}} (O_k^{\boldsymbol{x}})$.
\end{proof}

\begin{lemma}\label{thm:unfixed.critical.coalition}
  For $\boldsymbol{x}\in \widetilde{X}_k$ and $S\in 2^N \backslash \text{\emph{fix}}(\widetilde{X}_k)$, we have $\text{\emph{re}}_{\boldsymbol{x}} (S) \geq \min \{\text{\emph{re}}_{\boldsymbol{x}} (O_k^{\boldsymbol{x}}),\text{\emph{re}}_{\boldsymbol{x}} (R_k^{\boldsymbol{x}})\}$.
\end{lemma}
\begin{proof}
  Let $\boldsymbol{x}\in \widetilde{X}_k$ and $S\in 2^N \backslash \text{fix}(\widetilde{X}_k)$.
  For simplicity, let $\xi=\min \{\text{re}_{\boldsymbol{x}} (O_k^{\boldsymbol{x}}),\text{re}_{\boldsymbol{x}} (R_k^{\boldsymbol{x}})\}$.
  Recall that $\mathcal{P}_S$ is the set of a maximum disjoint $s$-$t$ paths in $D_S$ and $P_S$ is the arc set of all paths in $\mathcal{P}_S$.
  Let $\mathcal{P}^1_{S}$ be the set of paths $P\in \mathcal{P}_S$ with $J_P\backslash J_k\neq \emptyset$ and $P^1_S$ be the arc set of all paths in $\mathcal{P}^1_{S}$.
  Let $\mathcal{P}^2_{S}= \mathcal{P}_{S}\backslash \mathcal{P}^1_{S}$ and $P^2_S = P_S\backslash P^1_S$.
  Notice that $P^2_S \subseteq E_k$.
  We may assume that $\lvert \mathcal{P}^1_{S} \rvert \geq 1$.
  Otherwise, $S\in 2^N \backslash \text{fix}(\widetilde{X}_k)$ implies that there exists $e\in (S\backslash P_S)\backslash N_k$.
  Note that $(S\cap P_S)\cup \{e\}\in \mathcal{O}^e_k$.
  Then $\text{re}_{\boldsymbol{x}} (S)\geq \text{re}_{\boldsymbol{x}} \big((S\cap P_S)\cup \{e\}\big) \geq \text{re}_{\boldsymbol{x}} (O^{\boldsymbol{x}}_k)\geq \xi$.

  We apply induction on $\lvert \mathcal{P}^1_{S} \rvert$ to show $\text{re}_{\boldsymbol{x}} (S)\geq \xi$.

  First assume that $\lvert \mathcal{P}^1_{S} \rvert=1$.
  Let $P\in \mathcal{P}^1_{S}$ and $(u,v)\in J_P \backslash J_k$ be the first unfixed jump pair along $P$.
  If $\lvert J_P\backslash J_k \rvert=1$, then $S\cap P_S\in \mathcal{R}^{(u,v)}_k$ implies that $\text{re}_{\boldsymbol{x}} (S) \geq \text{re}_{\boldsymbol{x}} (S\cap P_S) \geq \text{re}_{\boldsymbol{x}} (R^{\boldsymbol{x}}_k)\geq \xi$.
  Hence we further assume that $\lvert J_P\backslash J_k \rvert \geq 2$ and $(u',v')$ is another unfixed jump pair in $J_P \backslash J_k$.
  By Lemma \ref{thm:jump.uniqueness}, there exists either a cycle $W$ or an $s$-$t$ path $W$ such that $J_W=\{(u,v)\}$.
  Note that $J_W \cap J_k=\emptyset$.
  We proceed by distinguishing two cases of $W$.

  Case $1$: $W$ is a cycle.
  By Lemma \ref{thm:potential.rep.path-cycle}, $x(N\cap W)=\phi_{\boldsymbol{x}} (u) - \phi_{\boldsymbol{x}} (v) \leq \sum_{(i,j)\in J_P} \big(\phi_{\boldsymbol{x}} (i) - \phi_{\boldsymbol{x}} (j) \big) - \big(\phi_{\boldsymbol{x}} (u') - \phi_{\boldsymbol{x}} (v') \big) = x(N\cap P)-1-\big(\phi_{\boldsymbol{x}} (u') - \phi_{\boldsymbol{x}} (v') \big)$.
  Since $(u,v)\in J_P\backslash J_k$, there exists $e^*\in (N\backslash N_k) \cap W$.
  Moreover, $(N\cap P^2_S)\cup \{e^*\}\in \mathcal{O}^{e^*}_k$.
  It follows that
  \begin{equation*}
    \begin{split}
      x(S)
      & = x(N\cap P) + x(N\cap P^2_S)\\
      & \geq 1 + \big(\phi_{\boldsymbol{x}} (u') - \phi_{\boldsymbol{x}} (v') \big) + x(N\cap W) + x(N\cap P^2_S)\\
      & \geq 1 + \big(\phi_{\boldsymbol{x}} (u') - \phi_{\boldsymbol{x}} (v') \big) + x(e^*) + x(N\cap P^2_S)\\
      & \geq 1 + \xi + (1+\xi) \lvert \mathcal{P}^2_S\rvert \\
      & = (1+\xi)\widetilde{\gamma}(S).
    \end{split}
  \end{equation*}

  Case $2$: $W$ is an $s$-$t$ path.
  By Lemma \ref{thm:potential.rep.path-cycle}, $x(N\cap W)=1+\big( \phi_{\boldsymbol{x}} (u) - \phi_{\boldsymbol{x}} (v)\big) \leq 1+ \sum_{(i,j)\in J_P} \big(\phi_{\boldsymbol{x}} (i) - \phi_{\boldsymbol{x}} (j) \big) - \big(\phi_{\boldsymbol{x}} (u') - \phi_{\boldsymbol{x}} (v') \big) = x(N\cap P)-\big(\phi_{\boldsymbol{x}} (u') - \phi_{\boldsymbol{x}} (v') \big)$.
  Let $W_{v-t}$ be the $v$-$t$ subpath of $W$.
  Let $w$ be the first common vertex of $W_{v-t}$ and a path $P'$ from $\mathcal{P}^2_S$.
  Note that $w$ could be $t$.
  Then $P_{s-u} * W_{u-w} * P'_{w-t}$ yields an $s$-$t$ path, denoted by $P^*$.
  Moreover, $\big( P^* \cup (P^2_S \backslash P') \big) \cap N\in \mathcal{R}^{(u,v)}_k$.
  It follows that
  \begin{equation*}
    \begin{split}
      x(S)
      & = x(N\cap P) + x(N\cap P') + x\big(N\cap (P^2_S \backslash P') \big)\\
      & = 1 + \sum_{(i,j)\in J_P} \big(\phi(i)-\phi(j)\big) + 1 + \sum_{(i,j)\in J_{P'}} \big(\phi(i)-\phi(j)\big) + x\big(N\cap (P^2_S \backslash P') \big)\\
      & \geq 1 + \sum_{(i,j)\in J_{P^*}} \big(\phi(i)-\phi(j)\big) + \big(\phi_{\boldsymbol{x}} (u') - \phi_{\boldsymbol{x}} (v') \big) + x\big(N\cap (P^2_S \backslash P') \big)  + 1\\
      & \geq x(N\cap P^*) +  \big(\phi_{\boldsymbol{x}} (u') - \phi_{\boldsymbol{x}} (v') \big) + x\big(N\cap (P^2_S \backslash P') \big)  + 1\\
      & \geq (1+\xi) \lvert \mathcal{P}^2_S \rvert + \xi +1 \\
      & = (1+\xi)\widetilde{\gamma}(S).
    \end{split}
  \end{equation*}

  Now assume that $\lvert \mathcal{P}^1_{S}\rvert=k+1$, where $k\geq 1$.
  Let $P\in \mathcal{P}^1_{S}$.
  By induction hypothesis, $\text{re}_{\boldsymbol{x}} \big(N\cap (P_S\backslash P) \big)\geq \xi$.
  It follows that
  \begin{equation*}
    \begin{split}
      x(S)
      & = x(N\cap P)+ x\big(N\cap (P_S\backslash P)\big) \\
      & \geq 1+\xi + (1+\xi) \big(\widetilde{\gamma}(S)-1\big)\\
      & = (1+\xi) \widetilde{\gamma}(S).
    \end{split}
  \end{equation*}

  Hence for $\boldsymbol{x}\in \widetilde{X}_k$ and $S\in 2^N \backslash \text{fix}(\widetilde{X}_k)$, we have $\text{re}_{\boldsymbol{x}} (S) \geq \min \{\text{re}_{\boldsymbol{x}} (O_k^{\boldsymbol{x}}),\text{re}_{\boldsymbol{x}} (R_k^{\boldsymbol{x}})\}$.
\end{proof}
 
  For any $(\boldsymbol{x},\epsilon)$ with $\boldsymbol{x}\in \widetilde{X}_k$, Lemma \ref{thm:unfixed.critical.coalition} implies that constraints \eqref{eq:separation.oracle.opt.k+1.c1} and \eqref{eq:separation.oracle.opt.k+1.c2} hold if and only if constraints \eqref{eq:auxiliary.nucleon.lp.c1} hold.
  Hence $\widetilde{LP}_{k+1}$ can be solved in polynomial time.
  Moreover, for $\boldsymbol{x}\in \widetilde{X}_k$, setting $\epsilon$ to $\epsilon_{k+1}$ in $\widetilde{LP}^{\boldsymbol{x}}_{k+1}$ yields a polynomial-time solvable linear program for the separation problem of $\widetilde{X}_{k+1}$.
 
\begin{lemma}
The separation problem of $\widetilde{X}_{k+1}$ can be solved in polynomial time. 
\end{lemma}

Since it takes at most $\lvert N\rvert$ iterations before $\widetilde{X}_{k+1}$ converges to the nucleon, we conclude that the nucleon of $\widetilde{\Gamma}_D$ can be computed in polynomial time.

\begin{lemma}\label{thm:nucleon.computation.polynomial}
$\eta (\widetilde{\Gamma}_D)$ can be computed in polynomial time.
\end{lemma}

\section{The approximate core and nucleon of $\Gamma_D$}
\label{sec:main.results}

Now we are ready to present our main results on $\Gamma_D$.
Theorem \ref{thm:approx.core} follows from Lemmas \ref{thm:multiplication.factor}, \ref{thm:nucleon.lp.1.equiv.opt.double.descr} and \ref{thm:core.v.rep}, which generalizes the result of Reijnierse et al. \cite{RMPT96}.

\begin{theorem}\label{thm:approx.core}
  Let $\Gamma_D=(N,\gamma)$ be a flow game defined on $D=(V,E;s,t)$ with player set $N\subseteq E$.
  Let $\epsilon^*=\max\{\epsilon: \mathcal{C}_{\epsilon}(\Gamma_D)\neq \emptyset\}$.
  Then $\epsilon^* = \frac{\sigma_E}{\sigma_N}-1$.
 Moreover, we have following characterizations for $\mathcal{C}_{\epsilon^*} (\Gamma_D)$.
  \begin{enumerate}[label*=\roman*\emph{)}]
    \item $\frac{\sigma_N}{\sigma_E} \mathcal{C}_{\epsilon^*} (\Gamma_D)= \{ \boldsymbol{x}\in \mathbb{R}^N_{\geq 0}: x(N)=\sigma_N; x(N\cap P)\geq 1, \forall P\in \mathscr{P}\}$.
    \item $\frac{\sigma_N}{\sigma_E} \mathcal{C}_{\epsilon^*} (\Gamma_D)$ is the convex hull of incidence vectors for minimum $s$-$t$ cuts constrained to $N$.
  \end{enumerate}
\end{theorem}

\begin{theorem}\label{thm:nulceon}
  Let $\Gamma_D=(N,\gamma)$ be a flow game defined on $D=(V,E;s,t)$ with player set $N\subseteq E$.
  Then $\eta (\Gamma_D)$ can be computed in polynomial time.
  Moreover, $\eta (\Gamma_D)=\mathcal{C}(\Gamma_D)$ if $\sigma_N = 1$ and $\eta (\Gamma_D)$ is a singleton otherwise.
\end{theorem}
\begin{proof}
  By Lemma \ref{thm:multiplication.factor}, we have $\eta (\Gamma_D)=\frac{\sigma_E}{\sigma_N}\eta (\widetilde{\Gamma}_D)$.
  Furthermore, Lemma \ref{thm:nucleon.computation.polynomial} establishes that $\eta (\Gamma_D)$ can be computed in polynomial time.

  First assume that $\sigma_N =1$.
  Clearly, $\gamma (S)\in \{0,1\}$ for $S\subseteq N$.
  Moreover, Reijnierse et al. \cite{RMPT96} showed that $\mathcal{C}(\Gamma_D)\neq \emptyset$ when $\sigma_N =1$. 
  Notice that for any $\boldsymbol{x}\in \mathcal{C}(\Gamma_D)$ and $S\subseteq N$, $\text{re}_{\boldsymbol{x}}(S)=0$ if $\gamma(S)=1$ and $\text{re}_{\boldsymbol{x}}(S)=+\infty$ otherwise.
  It follows that $\eta (\Gamma_D)=\mathcal{C}(\Gamma_D)$ in this case.

  Now assume that $\sigma_N \geq 2$.
  For any $e\in N$, there exists an $s$-$t$ path $P$ avoiding $e$.
  Hence $\gamma \big((N \cap P)\cup \{e\}\big) \geq \gamma(N \cap P)=1>0$.
  Notice that $\eta (\Gamma_D)$ fixes every coalition $S\subseteq N$ with $\gamma(S)>0$.
  Thus, both $(N \cap P)\cup \{e\}$ and $N\cap P$ are fixed in $\eta (\Gamma_D)$, implying that $e$ is also fixed in $\eta (\Gamma_D)$.
  It follows that every player is fixed in $\eta (\Gamma_D)$.
  Therefore, $\eta (\Gamma_D)$ is a singleton in this case.
\end{proof}



\bibliographystyle{habbrv}
\bibliography{reference}



\end{document}